\documentclass[preprint]{elsarticle} 

\usepackage[utf8]{inputenc} 

\usepackage{fullpage}

\usepackage{amssymb,amsthm,amsmath}\usepackage{url}\usepackage{graphicx}
\usepackage{multirow}
\newcommand{\cert}[1]{\mbox{\sc {certainty}}(#1)}

\newcommand{\eat}[1]{}

\newcommand{\squishlist}{
\begin{list}{$\bullet$}
{ \setlength{\itemsep}{0pt} \setlength{\parsep}{3pt}
\setlength{\topsep}{3pt} \setlength{\partopsep}{0pt}
\setlength{\leftmargin}{1.5em} \setlength{\labelwidth}{1em}
\setlength{\labelsep}{0.5em} } }

\newcommand{\squishend}{
\end{list}  }

\theoremstyle{remark}

\newtheorem{example}{Example}

\theoremstyle{plain}
\newtheorem{theorem}{Theorem}
\newtheorem{proposition}{Proposition}
\newtheorem{lemma}{Lemma}

\theoremstyle{definition}
\newtheorem{definition}{Definition}

\newcommand{\coNP}{\mathrm {coNP}}
\newcommand{\PTIME}{\mathrm {PTIME}}
\newcommand{\FO}{\mathrm {FO}}
\newcommand{\FP}{\mathrm {FP}}
\newcommand{\sP}{\mathrm {\#P}}

\newcommand{\Path}{\mbox {-{\sc Path}}}
\newcommand{\Cycle}{\mbox {-{\sc Cycle}}}

\begin{document}

\begin{frontmatter}

\title{Consistent Answers of Conjunctive Queries on Graphs
 \footnote{This work was supported by the project Handling Uncertainty in Data Intensive Applications, co-financed by the European Union (European Social Fund - ESF) and Greek national funds, through the Operational Program ``Education and Lifelong Learning", under the program THALES. Kolaitis is also partially supported by NSF Grant IIS-1217869.}}

\author[ntua]{Foto N. Afrati}
\ead{afrati@softlab.ntua.gr}
\author[cruz]{Phokion G. Kolaitis\corref{cor2}}
\ead{kolaitis@cs.ucsc.edu}
\author[ntua]{Angelos Vasilakopoulos }
\ead{avasila@central.ntua.gr}



\address[ntua]{National Technical University of Athens}
\address[cruz]{UC Santa Cruz and IBM Research - Almaden}

\cortext[cor2]{Corresponding author}


\begin{abstract}
During the past decade, there has been an extensive investigation of the computational complexity of  the consistent answers of Boolean conjunctive queries under primary key constraints. Much of this investigation has focused on self-join-free Boolean conjunctive queries. In this paper, we study  the consistent answers of Boolean conjunctive queries involving a single binary relation, i.e., we consider arbitrary Boolean conjunctive queries on
 directed graphs. In the presence of a single key constraint,  we show that for each such Boolean conjunctive query, either the problem of computing its consistent answers  is expressible in first-order logic, or it is polynomial-time solvable, but not expressible in  first-order logic.
  \eat{In contrast, in the presence of two key constraints,  either this problem is expressible in first-order logic or it is $\coNP$-complete.}
\end{abstract}


\end{frontmatter}

\section{Introduction}
\eat{
 The paper referenced below, is, in full, the:

Jef Wijsen: On the consistent rewriting of conjunctive queries under primary key constraints. Inf. Syst. 34(7): 578-601 (2009)

It is the only paper that Angelos found which treats a little bit self-joins.

Old Note by Angelos: Self-join query $R(\underbar{x},y) R(\underbar{z},y)$ does not satisfy conditions 2 and 4
of $[$Wijsen Inf.Sys.2009 Cor. 3$]$ so it is open if is FO or not. In this work, they define refiable queries.

This is a new note by Foto, I leave it here because I wrote it but I believe we only need to say that we prove the cycles are not FO-rewritable, hence they cannot fall in this case anyway.

Below we show that the queries we treat are not refiable (Definition 1 of citation above).

Moreover in the above citation, key-rooted queries (in other works called ``rooted queries'') are
FO rewrittable. Here our cycle queries are not -- as we show.

\begin{example}

Let the query be a $k$-cycle. Let the database be a $k$-cycle $a_1,..., a_k$ with additional edges
$(a_i, b_i)$ and $(b_i,a_{i+1}$ for each $i$ (where $k+1$ is replaced by  1). This database has the property:
for each constant $c_0$, there is a repair such that $c_0$ can be deleted from the repair and still the query
be true on what it remains. Hence no atom of the $k$-cycle query  is  refiable.

\end{example}
}

Database repairs and consistent query answering, introduced  in \cite{arenas99consistent}, provide a principled approach to the problem of managing inconsistency in databases and, in particular, to the problem of giving meaningful semantics to queries on an inconsistent database. If $\Sigma$ is a set of integrity constraints, then an \emph{inconsistent} database w.r.t.\ $\Sigma$ is a database instance $I$ that does not satisfy every constraint in $\Sigma$. A \emph{repair} of an inconsistent database instance $I$  is a database instance $J$ that satisfies every  constraint in $\Sigma$ and differs from $I$ in a ``minimal way".
  The \emph{consistent answers} of a query $q$ on  $I$  is the intersection $\bigcap \{q(J):
\mbox{$J$ is a repair of $I$}\}$. If $q$ is a Boolean query, then
computing the
consistent answers of $q$ is
the following decision problem, denoted by
  $\cert{q}$: given a database instance $I$, is $q(J)$ true on every repair $J$ of $I$?

There has been an extensive investigation of the algorithmic properties of consistent query answering for different classes of integrity constraints and different types of repairs (see \cite{Bertossi2011} for a survey). Much of the focus of this investigation has been on the consistent answers of conjunctive queries under primary key constraints and subset repairs.  Let $\bf S$ be a relational database schema such that every relation  in $\bf S$ has a single key. A \emph{subset} repair of a database instance $I$ over $\bf S$ is a maximal (under set inclusion) subinstance $J$ of $I$  that satisfies every key constraint of $\bf S$. It is easy to see that, in this scenario, for every Boolean conjunctive query $q$, we have that $\cert{q}$ is in $\coNP$.  It is also known that, depending on the query $q$ and the key constraints at hand, the actual computational complexity of $\cert{q}$ may vary from being $\coNP$-complete to being $\FO$-rewritable, i.e., there is a first-order expressible query $q'$ such that, for every database instance $I$, we have that $q$ is true on every subset repair of $I$ if and only if $q'$ is true on $I$.

The preceding state of affairs gave rise to a research program aiming to classify the computational complexity of $\cert{q}$, where $q$ is a Boolean conjunctive query under primary key constraints and subset repairs. After a sequence of partial results by several different researchers \cite{ChomickiM05,FuxmanM07,KolaitisP12,KoutrisS14,Wijsen09,Wijsen10,Wijsen12,Wijsen13} (see also \cite{Wijsen14} for a survey), a breakthrough  trichotomy result was recently announced by Koutris and Wijsen. Specifically,  in \cite{KoutrisW15},  Koutris and Wijsen showed that
 for every self-join-free Boolean conjunctive query $q$, one of the following three statements holds: (a) $\cert{q}$ is $\coNP$-complete; (b) $\cert{q}$ is in $\PTIME$; (c) $\cert{q}$ is $\FO$-rewritable. Moreover, there is an algorithm that, given such a query $q$, the algorithm determines which of these three statements holds for $\cert{q}$.

The hypothesis that the Boolean conjunctive queries considered have no self-joins plays a crucial role in the proof of the trichotomy theorem in \cite{KoutrisW15}. As a matter of fact, essentially all the earlier work on the classification of $\cert{q}$ is about self-join-free conjunctive queries, since most of the currently available techniques cannot handle the presence of self joins.  Two notable exceptions are $\coNP$-hardness results for specific Boolean conjunctive queries with self-joins in \cite{ChomickiM05} and a broad sufficient
condition
  for $\FO$-rewritability of Boolean conjunctive queries involving a single relation in \cite{Wijsen09}.

In this paper, we investigate the algorithmic aspects of $\cert{q}$, where $q$ is a Boolean conjunctive query over a single binary relation (hence, the query has self-joins, provided it has at least two atoms). In other words, we investigate the complexity of computing the consistent answers of arbitrary  Boolean conjunctive queries on directed graphs.  Our main focus is on the case in which there is a single key constraint, i.e., we focus on Boolean conjunctive queries over a single binary relation in which one of the attributes is a key. We show that if $q$ is such a conjunctive query, then either $\cert{q}$ is $\FO$-rewritable, or $\cert{q}$ is in $\PTIME$, but it  is not $\FO$-rewritable.  In addition, we characterize when each of these two cases occurs. More precisely, we first point out that every Boolean conjunctive query $q$  over a binary relation and with one of its attributes as a key is equivalent to either a  path query  or a collection of disjoint  cycles. We then show that if $q$ is a  path query or the query ``there is a self-loop", then $\cert{q}$ is $\FO$-rewritable; in contrast, if $q$ is a collection of disjoint  cycles each of length at least $2$, then $\cert{q}$ is in $\PTIME$, but it is not $\FO$-rewritable.

It should be pointed out that Maslowski and Wijsen \cite{MaslowskiW14} have established a dichotomy theorem for the problem $\#\cert{q}$ of counting the number of subset repairs satisfying a Boolean conjunctive query $q$ that may contain self-joins: for each such query $q$, either $\#\cert{q}$ is in $\FP$ (the class of polynomial-time solvable counting problems), or $\#\cert{q}$ is $\sP$-complete. When this result is applied to the case of Boolean conjunctive queries over a single binary relation, then it is not hard to verify  that $\#\cert{q}$ is in $\FP$ only when $q$ is equivalent to one of the queries, ``there is a path of length $1$",  ``there is a path of length 2", ``there is a self-loop"; for all other queries $q$, it turns out that $\#\cert{q}$ is $\sP$-complete. Thus, for Boolean conjunctive queries $q$  over a single binary relation, the dividing line between $\FO$-rewritability and $\PTIME$-computability for $\cert{q}$ is substantially different from the dividing line between membership in $\FP$ and $\sP$-completeness for $\#\cert{q}$.

\eat{
 Finally, we examine $\cert{q}$ when $q$ is a Boolean conjunctive query over a single binary relation such that each of its two attributes is a key.  In this setting, we show that if $q$ is a simple path, then $\cert{q}$ is $\coNP$-complete; in contrast, if $q$ is a disjoint union of cycles, then $\cert{q}$ is $\FO$-rewritable.}

 \section{Preliminaries}
 In general, a \emph{relational database schema} or, simply, a \emph{schema} is a finite collection $\bf R$ of relation symbols, each with an associated arity. Here, we will consider a  schema $\bf R$ consisting of a single binary relation $R$. We will review some of the basic notions of relational database theory for this particular setting.

A {\em relational database instance} over $\bf R$ or, simply, an \emph{instance} over $\bf R$ is a binary relation, which, for notational simplicity, we will also denote by  $R$. A \emph{fact} is an expression $R(a,b)$, where $a$ and $b$ are values such that $(a,b)\in R$. An instance over $\bf R$ can be thought of as a graph such that there is an edge from a node $a$ to a node $b$ if $R(a,b)$ is a fact of $R$.

We assume that the relation symbol $R$ has a single \emph{key} and that, actually, the first attribute of $R$ is a \emph{key}. A \emph{consistent instance} or a \emph{consistent graph} is a binary relation $R$ that  satisfies the key constraint, i.e., it does not contain two facts of the form $R(a,b)$ and $R(a,b')$ with $b \not = b'$.  An \emph{inconsistent instance} or an \emph{inconsistent graph} is a binary relation $R$ that violates the key constraint, i.e., $R$ contains two facts $R(a,b)$ and $R(a,b')$ with $b \not = b'$.

A \emph{subset repair} or, simply, a \emph{repair} of an instance $R$ is a maximal consistent sub-instance of $R$; in other words,  a repair of $R$  is  an instance $R'\subseteq R$ that satisfies the key constraint and such that there is no instance $R''$ with the property that $R' \subset R'' \subseteq R$ and $R''$ satisfies the key constraint.

 Let $q$ be a boolean query over the schema $\bf R$.
\squishlist
     \item $\cert{q}$ is the following decision problem: given an instance $R$, is $q$ true on every repair of $R$?
     \item $\#\cert{q}$ is the following counting problem: given an instance $R$, find
the number of repairs of $R$ that satisfy $q$.
\squishend

In this paper, we focus on conjunctive queries.  By definition,
 a \emph{conjunctive query} over a schema $\bf R$  is a first-order formula built from atomic formulas of $\bf R$ using conjunction  and
 existential quantification. If $\bf R$ consists of a single binary relation $R$, then  every conjunctive query is logically equivalent to an expression
 of the form
$q(\mathbf z)=\exists \mathbf{w}(R(\mathbf {x_1})\wedge ... \wedge R(\mathbf {x_m}))$, where
 each $\mathbf{x}_i$ is a pair of variables, $\mathbf{z}$ and $\mathbf{w}$ are tuples
  of variables, and the variables in $\mathbf {x_1},\ldots,\mathbf {x_m}$ appear in exactly one of $\mathbf{z}$ and $\mathbf{w}$. A \emph{boolean} conjunctive query is a conjunctive query in which
all variables are existentially quantified, i.e., $\mathbf{z}$ is the empty tuple.

 The {\em canonical database} of a  boolean conjunctive query $q$ is the instance $D^q$ obtained by viewing each variable in the query as a distinct value and each atom  as a fact of $D^q$. For example, if
    $q$ is the boolean conjunctive query $\exists x, y,z (R(x,y)\wedge R(y,z) \wedge R(z,x)$, then
    $D^q$ consists of the facts $R(x,y)$, $(y,z)$, $R(z,x)$.

  Two conjunctive queries $q$ and $q'$ are {\em equivalent} if for every  instance $R$, we have that
  $q(R)=q'(R)$.
   Starting with the work of Chandra and Merlin  \cite{ChandraM77}, there has been an extensive study of conjunctive-query equivalence and  minimization.
 A conjunctive query $q$ is {\em minimized} if there is no other  conjunctive query $q'$ which  is equivalent to $q$ and has
  fewer atoms in its definition than $q$ has.
  It is well known that every conjunctive query is equivalent to a unique (up to a renaming of the variables) minimized conjunctive query.  In terms of canonical databases, if we view the canonical database of a boolean conjunctive query $q$ as a graph $G$, then the canonical database of the minimized query $q'$ is the \emph{core} of the graph $G$, that is to say, a subgraph $G'$ of $G$ such that there is a homomorphism from $G$ to $G'$, but no homomorphism from $G$ to a proper subgraph $G''$ of $G$ (recall that a homomorphism from $G$ to $G'$ is a mapping $h$ from the nodes of $G$ to the nodes of $G'$ such that if $(u,v)$ is an edge of $G$, then $(h(u),h(v))$ is an edge of $G'$).

Two conjunctive queries $q$ and $q'$ are {\em equivalent under the key constraint} of the binary relation $R$
    if for every consistent instance $R$,  we have that
  $q(R)=q'(R)$.  Clearly, if $q$ and $q'$ are equivalent under the key constraint of $R$, then $\cert{q}$ coincides with $\cert{q'}$; similarly, $\#\cert{q}$ coincides with $\#\cert{q'}$.
   Conjunctive query  equivalence under various integrity constraints has been investigated in  various settings in the past (see, e.g.,
   \cite{CDGL98,AhoSU79,JohnsonK84}).
%

 \section{Conjunctive-Query Equivalence under a Key Constraint}

 We will analyze conjunctive-query equivalence under the key constraint of the binary relation symbol $R$. For example, consider the conjunctive query
$\exists x, y, z(R(x,y)\wedge R(x,z) \wedge R(y,z))$, where the first attribute of  $R$ is a key. Observe that, if this query
evaluates to true on a consistent instance, then the variables $y$ and $z$ must be instantiated to the same value.
Hence, under the key constraint, this query is equivalent to
$\exists x, y, z(R(x,y)\wedge R(x,y) \wedge R(y,y))$,
which, in turn,  is equivalent to $\exists x, y(R(x,y)\wedge R(y,y))$.
 We shall show that every boolean conjunctive query is equivalent under  the key constraint to a  boolean conjunctive query that has a rather simple form. As a first step, we analyze the structure of consistent instances.

\begin{proposition} \label{instance-structure-prop}
\label{thm-cons} Let $\bf R$ be a schema consisting of a single binary relation symbol with the first attribute as key. An instance $R$ is consistent if and only if $R$, when viewed as a graph, is the union of
a forest of trees oriented from the leaves to the root and of simple cycles whose nodes either are not in the forest or are roots of some trees of the forest.
\end{proposition}

\begin{proof} The direction from right to left is obvious.
 For the other direction, suppose that $R$ is a consistent instance.
 Let $C$ be a simple cycle of $R$.  If $v$ is a node on $c$, then the only
 outgoing edge from $v$ is the edge that goes to the next node on $C$ (otherwise, $R$ is inconsistent). Thus, there are no edges from a node of $C$ to some node outside $C$.
 Moreover, a directed acyclic graph is a  consistent instance  if and only if it is a forest of trees oriented from the leaves to the root.
It follows that $R$  consists  of a set of disjoint simple cycles and a set  of disjoint trees oriented from the leaves to the root, where the root of such a tree may possibly also be on one of the cycles.
\end{proof}

Let us return to the boolean conjunctive query
$\exists x, y, z(R(x,y)\wedge R(x,z) \wedge R(y,z))$. As seen earlier, this query is equivalent under the key constraint to the boolean conjunctive query $\exists x, y(R(x,y)\wedge R(y,y))$. The canonical database of the latter consists of the facts $R(x,y)$ and $R(y,y)$, hence its core consists of just the fact $R(y,y)$. It follows that
$\exists x, y, z(R(x,y)\wedge R(x,z) \wedge R(y,z))$ is equivalent under the key constraint to the existence-of-a-self-loop query
$\exists y R(y,y)$. It turns out that, by first applying repeatedly the key constraint and then minimizing, every boolean conjunctive query is equivalent to one that has a simple structure.

\begin{definition} \label{simple-defn}
Let $\bf R$ be a schema consisting of a single binary relation symbol.
\squishlist
\item For every $n\geq 2$, we write $n\Path$ to denote the  boolean conjunctive query that asserts the existence of a  path of length $n$, i.e., $n\Path$ is of the form.
$\exists x_1,\ldots,x_n (R(x_1,x_2)\wedge \cdots \wedge R(x_{n-1},x_n))$.

 We say that a boolean conjunctive query is a \emph{simple path} query  if it is the  $n\Path$ query, for some $n\geq 2$.
\item For every $n\geq 1$, we write $n\Cycle$ to denote the  boolean conjunctive query that asserts the existence of a simple  cycle  of length $n$, i.e., $n\Cycle$ is of the form
    $\exists x_1,\ldots,x_n (R(x_1,x_2)\wedge \cdots \wedge R(x_{n},x_1))$.

    We say that a boolean conjunctive query is a \emph{cycle} query  if it is the  $n\Cycle$ query, for some $n\geq 1$. We also say that a boolean conjunctive query is a \emph{disjoint collection of simple cycles} if it is the conjunction of simple cycle queries with no variables in common.   
    
    For example, the query
    $\exists x_1,\ldots,x_5(R(x_1,x_2) \wedge R(x_2,x_1) \wedge R(x_3,x_4)\wedge R(x_4,x_5)\wedge R(x_5,x_3))$ is the disjoint collection of the $2\Cycle$ query and the $3\Cycle$ query.

\squishend
\end{definition}

\begin{theorem} \label{query-structure-thm}
Let $\bf R$ be a schema consisting of a single binary relation symbol with the first attribute as key. Every
 boolean conjunctive query over $\bf R$ is equivalent under the key constraint either to a  path query or to a query that is a disjoint collection of  cycles such that the length of each cycle in the collection does not divide the length of any other cycle in the collection. Moreover, there is a polynomial-time algorithm that, given a boolean conjunctive query over $\bf R$, decides which of these two cases holds.
\end{theorem}

\begin{proof} Let $q$ be a boolean conjunctive query over $\bf R$.
First, form the finest partition of the variables of $q$ such that
if we replace all variables in a single part of the partition with a fresh variable, then the
canonical database $D^p$ of the resulting boolean conjunctive query $p$ is a consistent instance.
Intuitively, this is achieved by considering all atoms with the same variable, say $x$, in the first attribute and by replacing all occurrences of variables that appear in the second attribute of these atoms with the same fresh variable $x_f$.  Clearly, $q$ is equivalent under the key constraint to $p$.
Since the canonical database $D^p$ of $p$ is a consistent instance, the preceding Proposition \ref{instance-structure-prop} implies that $D^p$ is  the union of
a forest of trees oriented from the leaves to the root and of simple cycles whose nodes either are not in the forest or are roots of some trees of the forest. If $D^p$ is actually an acyclic graph, then the core of $D^p$ is a simple path, hence $q$ is equivalent under the key constraint to a  path query. If $D^p$ contains at least one cycle, then its core is a collection of disjoint cycles such that the length of each cycle in the collection does not divide the length of any other cycle in the collection (note that every tree can be homomorphically mapped to any cycle).  It follows that, in this case, $q$ is equivalent under the key constraint to a disjoint collection of cycles such that the length of each cycle in the collection does not divide the length of any other cycle in the collection.
\end{proof}

\section{First-Order Rewritability}
\label{FO-sec}

 Let $q$ be a boolean conjunctive query over some relational schema $\bf S$.  We say that
$\cert{q}$ is \emph{$\FO$-rewritable} if there is a boolean first-order query over $\bf S$ such that for every instance $I$ of $\bf S$, we have that every repair of $I$ satisfies $q$ if and only if $I$ satisfies $q'$.
 For self-join-free conjunctive queries,  a systematic study of when $\cert{q}$ is  $\FO$-rewritable  was carried out first by Fuxman and Miller \cite{FuxmanM07} and then by Wijsen \cite{Wijsen12}. In this section, we
  obtain the following characterization of  $\FO$-rewritability of boolean conjunctive queries over a schema consisting of a single binary relation with a single key constraint.

\begin{theorem} \label{FO-thm}
 Let $\bf R$ be a schema consisting of a single binary relation symbol with the first attribute as the key.
If $q$ is boolean conjunctive query over $\bf R$, then the following two statements are equivalent.
\squishlist
\item  $\cert{q}$ is $\FO$-rewritable.
\item  $q$ is equivalent under the key constraint to the $1\Cycle$ query or to a path query.
\squishend
\end{theorem}

We will first show that if $q$ is the $1\Cycle$ query or a path query, then $\cert{q}$ is $\FO$-rewritable.
For the  $1\Cycle$ query $\exists xR(x,x)$, it is easy to verify that the sentence $\exists x(R(x,x) \wedge
\forall y (x\not = y \rightarrow \neg R(x,y)))$ is a first-order rewriting of $\cert{1\Cycle}$. Indeed, if an instance $R$ satisfies the preceding sentence, then there is a node $a$ such that the only edge coming out of $a$ is the self-loop $R(a,a)$. Hence, every repair of $R$ must contain the fact $R(a,a)$, which means that every repair of $R$ contains a self-loop.  Conversely, if every repair of $R$ contains a self-loop, then $R$ must satisfy
the sentence $\exists R(x,x) \wedge
\forall y (x\not = y \rightarrow \neg R(x,y)))$, since, otherwise, we could construct a repair $R'$ of $R$ that contains no self-loops, since, for every node $a$ such that $R(a,a)$ is a fact of $R$, there is a node $b\not =a $ such that $R(a,b)$ is a fact of $R$, and we can form the desired repair $R'$ by putting such facts $R(a,b')$ in it.

As regards to path queries,
 note that Fuxman and Miller \cite{FuxmanM07} identified a class,   called  ${\mathcal C}_{forest}$, of self-join-free conjunctive queries and showed that if $q$ is a query in ${\mathcal C}_{forest}$, then $\cert{q}$ is $\FO$-rewritable. The class ${\mathcal C}_{forest}$ includes as a member every query $q_n$, $n\geq 2$, of the form $\exists x_1 \ldots x_n(S_1(x_1,x_2)\wedge S_2(x_2,x_3) \wedge \cdots \wedge S_{n-1}(x_{n-1},x_n))$, where  the relation symbols $S_j$  are distinct. In general, the first-order rewriting algorithm for queries in ${\mathcal C}_{forest}$ fails if it is applied to conjunctive queries with self-joins. It can be shown, however, that this algorithm produces a correct first-order rewriting when applied to the queries $n\Path$, $n\geq 2$. Here, we give a direct proof of this result.

\begin{theorem} \label{Path-thm}
Let $\bf R$ be a schema consisting of a single binary relation symbol with the first attribute as the key.
If $q$ is a path query, then $\cert{q}$ is $\FO$-rewritable.
\end{theorem}

\begin{proof}
We begin by giving the first-order rewriting of $\cert{2\Path}$. Let $\psi_2$ be the  first-order sentence
$$\exists x,y,z [ R(x,y) \wedge R(y,z) \wedge \forall y(  R(x,y)\rightarrow \exists zR(y,z)  )           ].$$
We claim that $\psi_2$ is a first-order rewriting of $\cert{2\Path}$.
Intuitively, $\psi_2$ asserts that
  there is a path of length $2$  in the database and, moreover,  whenever
we replace in some repair the first edge of this path with another edge whose endpoint is a node $u$, then there is an edge starting from this node $u$. This ensures that every repair contains a path of length $2$.

 More formally, suppose first that an instance $R$ satisfies $\psi_2$, and that $R'$ is a repair of $R$. Then there are nodes $a$, $b$, $c$ such that $R(a,b)$ and $R(b,c)$ are facts of $R$. It follows that $R'$ must contain a fact of the form $R(a,b')$ for some node $b'$. Since $R$ satisfies $\psi_2$, there is a node $c'$ such that $R(b',c')$ is a fact of $R$. Consequently $R'$ must contain a fact of the form $R(b',c'')$, hence $R'$ contains the path $R(a,b')$, $R(b',c'')$. Next, assume that $R$ does not satisfy $\psi_2$. We will show how to construct a repair $R'$ of $R$ that contains no path of length $2$. If $a$ is a node for which there is a fact $R(a,b)$ of $R$ such that there is no fact of the form $R(b,c)$ in $R$, then we pick one such $b$ and put $R(a,b)$ in $R'$.
 Since $\psi_2$ is false on $R$, if  $a$, $b$ and $c$
 are three nodes
 such that
$R(a,b)$ and $R(b,c)$ are facts of $R$, then there is a $b'$ such that $R(a,b')$ is a fact of $R$,  but there is no node $c'$ such that $R(b',c')$ is a fact of $R$.  In this case, we add $R(a,b')$ to the repair $R'$. We continue doing the same  for all
nodes $a$ that are the beginning of a path of length $2$ in $R$. This construction produces a repair $R'$ of $R$ that contains no path of length $2$.

Next, let $\psi_3$ be the first-order sentence
$$\exists x,y,z,w [ R(x,y) \wedge R(y,z) \wedge R(z,w) \wedge \forall y (  R(x,y) \rightarrow $$
$$   \exists z,w [ R(y,z) \wedge R(z,w) \wedge \forall w'(  R(y,w')\rightarrow \exists z'R(w',z')  )           ]$$
Observe that the subformula of $\psi_3$ shown in the second row is essentially the formula $\psi_2$, except that
an existential quantifier is missing in the front.

It is not hard to verify that $\psi_3$ is a first-order rewriting of $\cert{3\Path}$.
 The intuition is analogous to that for $\psi_2$, namely, $\psi_3$ asserts that there is a path of length $3$ in the database and, moreover, when we replace in some repair the first edge with another edge whose endpoint is a node $u$, then there is a path of length $2$ starting from the node $u$. This ensures that every repair contains a path of length $3$.

A first-order rewriting $\psi_n$ of $\cert{n\Path}$, for $n> 3$, can be obtained via an inductive definition that is similar to the way $\psi_3$ was obtained from $\psi_2$. Specifically, the first part asserts the existence of a path of length $n$ and the rest of $\psi_n$ is essentially the first-order rewriting $\psi_{n-1}$ of $\cert{(n-1)\Path}$, except that an existential quantifier is missing in the front.
\end{proof}

Next, we focus on queries that are a disjoint collection of cycles each of which has length at least $2$.

\begin{theorem} \label{Cycle-thm}
Let $\bf R$ be a schema consisting of a single binary relation symbol with the first attribute as the key.
 Assume that  $q$  is a disjoint collections of cycles such that each cycle in the collection has length at least $2$, and
the length of each cycle in the collection does not divide the length of any other cycle in the collection.
Then $\cert{q}$ is not $\FO$-rewritable.
\end{theorem}
\begin{proof}
We first prove the result for the case in which the query is a single cycle of length at least $2$; in other words, we will show that if $n\geq 2$, then $\cert{n\Cycle}$ is not $\FO$-rewritable. The proof uses the technique of Ehrenfeucht-Fra\"iss\'ee games (see \cite{Libkin04} for an exposition). For concreteness, we provide the details for $\cert{2\Cycle}$ and for $\cert{3\Cycle}$; the generalization to cycles of bigger length will become clear for the constructions in these two cases.

For the $2\Cycle$ query, let $D_1$ and $D_2$ be the database instances depicted in Figure \ref{fig:cycle2}.

\begin{figure} \begin{center}
\includegraphics[width=0.59\textwidth]{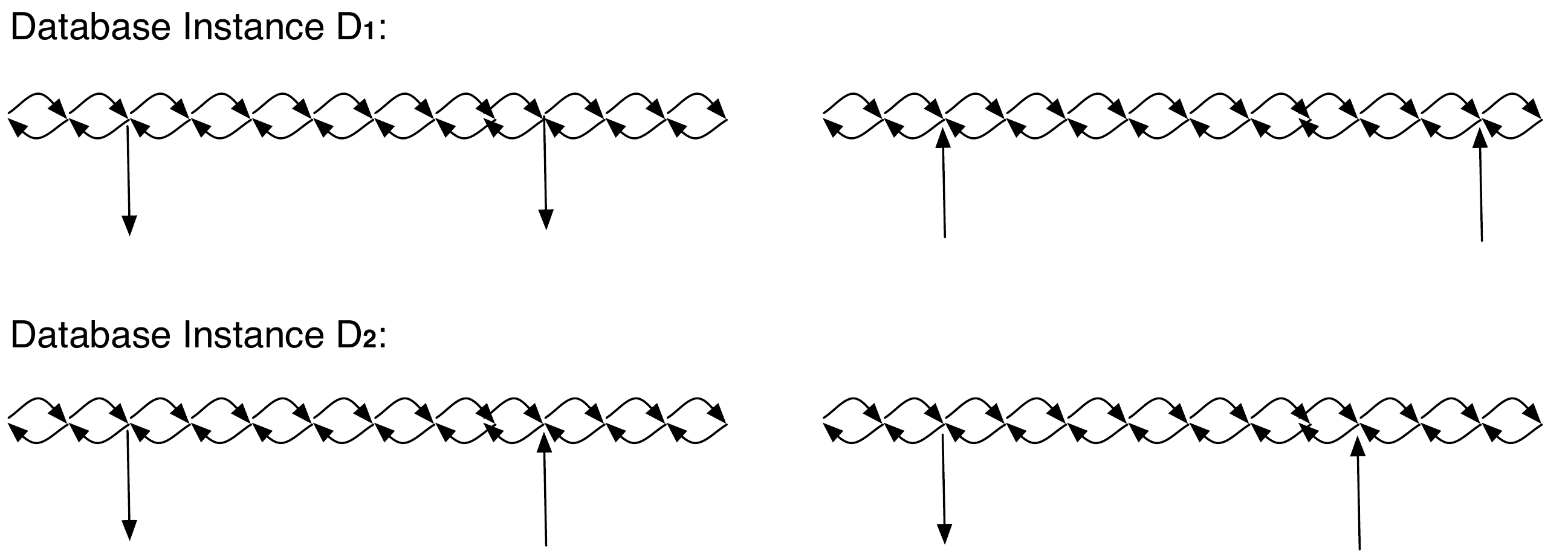}
\caption{Databases $D_1$ and $D_2$}\label{fig:cycle2}
\end{center}
\end{figure}

\squishlist
\item {\em Database Instance $D_1$}: There are two disjoint simple paths of ``double" edges, say $R(u,v)$ and $R(v,u)$, each of which forms a $2$-cycle.
For the first path, there are two simple edges going out from two nodes that ``far apart" and also ``far" from the point the endpoints of the path.
For the second path,
 there are two simple edges entering at two nodes on the path that are  ``far apart".

\item {\em Database Instance $D_2$}:
As in $D_1$, there are two disjoint simple paths of ``double" edges, say $R(u,v)$ and $R(v,u)$, each of which forms a $2$-cycle.
For the first path,  there is one simple edge entering and one simple
edge going out at nodes that are ``far apart", and also ``far" from the endpoints of the path.
The second path is a copy of the first path of $D_2$.
\squishend

We claim that every repair of $D_1$ satisfies the $2\Cycle$ query, while there is a repair of $D_2$ on which the $2\Cycle$ query is false. To see this, consider first a simple path of ``double" edges (with no ingoing or outgoing simple edges at some node). In such a path, all nodes have outdegree $2$, except for the two endpoints of the path. Thus, the edges emanating from the two endpoints must be included in every repair of the path. From this, it follows that every repair must contain a cycle of length $2$, since, if one tries to build a repair that avoids $2$-cycles, then one ends up including a cycle of length $2$ at one of the two endpoints. The situation remains the same if we augment the path with simple ingoing edges. From this, it follows that every repair of $D_2$ contains a $2$-cycle in the right component of $D_2$. If, however, we augment the path with at least one simple outgoing edge, then we can construct a repair of the path that has no $2$-cycles. Since both components of $D_2$ have an outgoing simple edge, it follows that $D_2$ has a repair that has no $2$-cycles.

Let us call a node in $D_1$ or in $D_2$ \emph{special} if, in addition to the edges of the $2$-cycle, it has an ingoing or outgoing simple edge. Fix a positive integer $m$ and consider the $m$-move Ehrenfeucht-Fra\"iss\'e game on two instances that have the same shape as $D_1$ and $D_2$. If the distance between the two special nodes in each component is large enough, then it is easy to see that the Duplicator wins the $m$-move Ehrenfeucht-Fra\"iss\'e game on these two instances, because when the Spoiler plays close to a special node in one of the instances, then the Duplicator can play on a similar node (i.e., with an outgoing or ingoing extra simple edge) in the other instance. Consequently, $\cert{2\Cycle}$ is not $\FO$-rewritable.

For the $3\Cycle$ query,
consider a sequence of consecutive $3$-cycles $R(a,b_1)$, $R(b_1,c_1)$, $R(c_1,a)$;  $R(c_1,b_2)$, $R(b_2,c_2)$,
$R(c_2,c_1)$; $R(c_2,b_3)$,   $R(b_3,c_3)$, $R(c_3,c_2)$, and so on, until the last $3$-cycle, say, $R(c_m,b_{m+1})$, $R(b_{m+1},a')$, $R(a',c_m)$. In this instance, only the nodes $c_1,c_2,\ldots$ have degree bigger than one, so all other edges must be in every repair. Consider a repair of this instance. If it contains $R(c_1,a)$ or some edge $R(c_i,c_{i-1})$, then the repair contains a $3$-cycle. The best chance to eliminate  all $3$-cycles in a repair is to remove $R(c_1,a), R(c_2,c_1), \ldots, R(c_m,c_{m-1})$. But then we must keep the edges $R(c_m,b_{m+1})$, $R(b_{m+1},a')$, $R(a',c_m)$, which form a $3$-cycle.

Now, if we have incoming edges to some of the $c_i$'s, it will still be the case that every repair contains a $3$-cycle. On the other hand, even a single outgoing edge to some $c_i$ allows to get a repair that has no $3$-cycles by keeping this outgoing edge from $c_i$ and removing the edges $R(c_i,c_{i-1}), R(c_i,b_{i+1})$, and ultimately removing the edge $R(c_m,b_{m+1})$, thus eliminating all $3$-cycles in the process.

The rest then is as for the instances with the $2\Cycle$ query.  We form two instances $D_1$ and $D_2$ with two chains of triangles in each, and ingoing and outgoing edges to some nodes $c_i$ and $c_j$ as in Figure \ref{fig:cycle2}, and use Ehrenfeucht-Fra\"iss\'e games to conclude that $\cert{3\Cycle}$ is not $\FO$-rewritable.


Finally,  we need to consider disjoint collections of cycles.
  For concreteness, let $q$ be a query made up of  three cycles  $C_1$,
$C_2$, $C_3$ such that the length of each of these cycles  does not divide the length of anyone of the other cycles (in particular, the length of each cycle is at least $2$). Consider the instances $D_1$ and $D_2$ for the query associated with the cycle $C_1$, and form the instances $D'_1\oplus C_2\oplus C_3$ and $D'_2= D_2\oplus C_2\oplus C_3$ obtained by forming the disjoint union of $D_1$, $C_2$, $C_3$, and the disjoint union of $D_2$, $C_2$, $C_3$.
It it not hard to show that every repair of $D'_1$ satisfies the query $q$, while there is a repair of $D'_2$ that it does not. Moreover, for every positive integer $m$, we can construct instances that have the shape of $D'_1$ and $D'_2$, and are such that the Duplicator wins the $m$-move Ehrenfeucht-Fra\"iss\'e game played on these instances.  Consequently, $\cert{q}$ is not $\FO$-rewritable.
\end{proof}

Theorem \ref{FO-thm} now follows from Theorems \ref{query-structure-thm},  \ref{Path-thm}, \ref{Cycle-thm}, and the earlier remarks about the $1\Cycle$ query.

\section{Polynomial-Time Computability}
\label{algo-sec}

Let $\bf R$ be a schema consisting of a single binary relation symbol with the first attribute as the key. In this section, we show that if $q$ is a boolean conjunctive query over $\bf R$, then $\cert{q}$ is in $\PTIME$. Clearly, if
$\cert{q}$ is $\FO$-rewritable, then $\cert{q}$ is in $\PTIME$. Thus,
in view of Theorems \ref{query-structure-thm} and \ref{FO-thm}, it suffices to show that $\cert{q}$ is in $\PTIME$ whenever $q$ is a disjoint collection of cycles such that the length of each cycle in the collection is at least $2$ and it does not divide the length of any other cycle in the collection. This will be accomplished in a series of steps that build to the main result.

Before we proceed, we need to recall  the following facts from graph theory, which will be useful in some of the proofs. If $G$ is a graph, then the strongly connected components  of  $G$ form a partition of the set of nodes of $G$.
If each strongly connected component is contracted to a single node, the resulting graph is a directed acyclic graph, called the {\em condensation} graph of $G$.  The strongly connected components that  are contracted into sink nodes in the condensation graph of $G$ are called {\em sink} strongly connected components.

We start with proving that $\cert{n\Cycle}$, $n \geq 2$, is in $\PTIME$.  For this,  we need some lemmas in which we always assume that we have a schema consisting of a single binary relation with the first attribute as a key.

\begin{lemma}
\label{lemma-0}  If $D$ is a repair of an instance $R$, then every cycle in $D$ is simple.
\end{lemma}

\begin{proof} Assume that $D$ is a repair of $R$ such that $D$ contains a cycle $C$ that is not simple. This means that $C$ contains a node $u$ with outgoing edges to two distinct nodes $v_1$ and $v_2$. Consequently, $D$ contains the facts $R(u,v_1)$ and $R(u,v_2)$, hence $D$ violates the key constraint, which is a contradiction.
\end{proof}

\begin{lemma}
\label{lemma-1}
\label{algo-lemma}
If $D$ is a repair of an instance $R$ and if $S$ is a sink strongly connected component of $R$, then
the intersection $D\cap S$ contains a simple cycle.
\end{lemma}
\begin{proof}
Let $u$ be a node in  $D\cap S$. Since $S$ is a sink strongly connected component of $R$, every edge outgoing from $u$ must lead to a node in $S$. Therefore, there is a node $v$ such that $R(u,v)$ is a fact of $D\cap S$.
By the same reasoning, there a node $w$ such that $R(v,w)$ is in $D\cap S$, and so on; thus, we obtain a path in which every edge is in $D\cap S$. Since $D\cap S$ is finite, at some point
we will encounter for the first time a node that has been earlier in the path, hence  $D\cap S$ contains a simple cycle.
\end{proof}

\begin{lemma}
\label{lemma-3} If $R$ is an instance,  $S$ is a strongly connected component of $R$, and $C$ is a simple cycle in $S$,  then
  there is a repair $D$ of $R$  such that
$D\cap S$
   contains the simple cycle $C$ and no other simple cycle.
\end{lemma}

\begin{proof}
First, we  build a repair $D'$ that contains the simple cycle $C$. For this, we
 include the simple cycle $C$ in the repair and then we keep adding edges until no more edges can
be added while, at the same time,  satisfying the key constraint of $R$.

 Second, we build  a repair $D$ such that  $D\cap S$ contains the simple cycle $C$ and no other simple cycle.
 For this, we start with the repair $D_1$ in the previous step for which we have that $D_1 \cap S$ contains the simple cycle $C$. Suppose that $D_1 \cap S$ contains another simple cycle $C'$. Note that $C$ and $C'$ have no nodes in common, since, by Lemma \ref{lemma-0}, every cycle in $D_1$ is simple. We construct another repair
$D_2$ by ``breaking" $C'$ as follows.  The strongly connected component  $S$ contains  a node $u$  on $C'$   such that there is a path $p$  from $u$ to a node
of $C$. We delete the outgoing edge from $u$ that belongs to  $C'$ and add the edge $e=(u,v)$ of $u$ that is on the path $p$.
If the  newly added edge $e$ has as endpoint a node on $C$, then it does not create a new cycle.  Otherwise, it may create a new cycle $C^{''}$; however,  there is now a path from $v$ to a node of
$C$ that is shorter than $p$ (in fact, this path is obtained from $p$ by deleting the edge  $e$); this way,
 we continue breaking
the next cycle (if there is one) until no cycles other than $C$ are left.
\end{proof}

\begin{lemma}
\label{lemma-2} Every instance $R$ has a repair $D$ in which the only simple cycles are in sink strongly connected components.
\end{lemma}
\begin{proof}
%

Towards building the desired repair $D$, we start by
building a repair $D'$ that has a simple cycle in  each non-sink strongly connected component.
Actually, for every strongly connected component $S$, we can choose any cycle we want  in $D'\cap S$ (as per Lemma~\ref{lemma-3}), so we choose a cycle that has a special node $u$ such that there is an edge from $u$ to
 a strongly connected component at the next level of the condensation graph of $R$ (recall that the condensation graph is a directed acyclic graph).
We now build $D$ from $D'$ by ``breaking" cycles as follows.  For every strongly connected component $S$, we take the special node $u$ in $D'\cap S$,
and we remove the edge outgoing from $u$ in the cycle and add the edge that goes from $u$ to the next level. This new edge does not create a cycle because inter-level edges do not belong to any cycle.
\end{proof}

We are now ready to state and prove the main technical result of this section.

\begin{theorem}
\label{thm1} Consider the $n\Cycle$ query, where $n\geq 2$. Then the following statements are true.
\squishlist
\item Every repair of an instance $R$ satisfies the $n\Cycle$ query if and only if there is a sink strongly component $S$ of $R$ such that every simple cycle of $S$ is a homomorphic image of an $n$-cycle.
    \item $\cert{n\Cycle}$ is in $\PTIME$.
    \squishend
\end{theorem}
\begin{proof}
For the first part of the theorem, recall that, by Lemma~\ref{lemma-0}, the only images of an $n$-cycle that are consistent instances are simple cycles.
 For the ``if'' direction, suppose that there is a sink strongly connected component such that every simple cycle is a homomorphic image of an $n$-cycle. By Lemma~\ref{lemma-1}, every repair contains a cycle from each sink strongly connected component, hence   every repair satisfies the $n\Cycle$ query.
 For the ``only if" direction, if there is no sink strongly connected component with the property that all its simple cycles are homomorphic images of an $n$-cycle, then, by Lemmas \ref{lemma-3} and  \ref{lemma-2}, we can build a repair $D$ that contains no homomorphic image of an $n$-cycle, hence $D$ does not satisfy the $n\Cycle$ query.

For the second part of the theorem, we need to prove that the property in the first part of the theorem can be checked in polynomial time.  Consider the following algorithm:

For each sink strongly connected component, do:
\squishlist
\item [1.]  To check that there is no simple cycle with more than $n$ nodes:
\squishlist
\item[(a)]   Examine each $n$-tuple of nodes $(a_1,...,a_n)$ and test whether they form a simple path.

\item[(b)] If they do, check whether there is a path from $a_n$ to $a_1$ that does not contain any of the other nodes $a_2,...,a_{k-1}$.
\squishend

Since $n$ is a fixed number, this last check can be done in polynomial time and, in fact, even in Datalog with inequalities $\neq$. If such a disjoint path exists, then there is a simple cycle with more than  $n$ nodes through $(a_1,...,a_n)$. Otherwise, there is no simple cycle with more than $n$ nodes through $(a_1,...,a_n)$.

\item [2.] For cycles with fewer than $n$ nodes, we check exhaustively in polynomial time all combinations of $k<n$ nodes to see whether they form a cycle which is
not a homomorphic image of $n$-cycle.
\squishend
Clearly, this algorithm runs in time bounded by a polynomial in the size of $R$ (the degree of the polynomial depends on $n$, which is a fixed number).
\end{proof}

Finally, we consider  queries that are   a disjoint collection of simple cycles.

\begin{theorem}
\label{thm2} Let $q$ be a boolean conjunctive query that is a disjoint collection of cycles $C_1,\ldots C_m$  such that the length of each cycle in the collection is at least $2$ and it does not divide the length of any other cycle in the collection.
Then the following statements are true.
\squishlist
\item
Every repair of an instance $R$ satisfies the query $q$  if and only if  for every cycle $C_i$, $1\leq i\leq m$,  in $q$,  there exists a sink strongly connected component of $R$ such that all its
simple cycles are homomorphic images of $C_i$.
\item $\cert{q}$ is in $\PTIME$.
\squishend
\end{theorem}

\begin{proof}
 For the ``if'' direction of the first part of the theorem,
since every repair $D$ of  $R$  contains a cycle from each sink strongly connected component, we have that $D$ contains
 a homomorphic image of the cycle $C_i$, for each $i\leq m$.
For the ``only if" direction of the first part of the theorem,
suppose that there is a cycle $C_i$ such that no sink strongly connected component contains a homomorphic image of it.  Lemmas \ref{lemma-3}  and \ref{lemma-2} yields a repair that does not satisfy the query $q$.

For the second part of the theorem, we need to check that the condition in the first part of the theorem can be checked in polynomial time. This is an easy consequence of the second part of Theorem \ref{thm1}.
\end{proof}

We conclude the paper by combining Theorems \ref{query-structure-thm}, \ref{FO-thm}, and \ref{thm2} into a  single result.

\begin{theorem} \label{summary-thm} Let $\bf R$ be a relational schema consisting of a single binary relation $R$ in which the first attribute is the key.
If $q$ is a boolean conjunctive query over $\bf R$, then $\cert{q}$ is in $\PTIME$. Moreover, exactly one of the following two possibilities hold:
\squishlist
\item The query $q$ is $\FO$-rewritable, and it is equivalent under the key constraint to the $1\Cycle$ query or to a path query.
    \item The query $q$ is not $\FO$-rewritable, and it is equivalent under the key constraint to a disjoint collection of cycles such that the length of each cycle is at least $2$ and it does not divide the length of any other cycle in the collection.
        \squishend
        \end{theorem}

\bibliographystyle{elsarticle-num}
\bibliography{cqabib,repairs}

\end{document}